\documentclass[%
reprint,
superscriptaddress,
showpacs,%preprintnumbers,
 amsmath,amssymb,
 aps,
 pra,
]{revtex4-1}

\usepackage{graphicx}% Include figure files
\usepackage{dcolumn}% Align table columns on decimal point
\usepackage{bm}% bold math
\usepackage{todonotes}
\usepackage[caption=false]{subfig}
\usepackage{xcolor}
\usepackage{float}
\usepackage{placeins}
 \usepackage{subfig}
\usepackage[normalem]{ulem}
\usepackage{amsmath}
\usepackage{hyperref}
\usepackage{bm}
\usepackage{dsfont}
\usepackage{amsthm}
\usepackage{verbatim} 
\usepackage{qcircuit}
\usepackage{circuitikz}
\usepackage{tabularx}
\usepackage{colortbl}
\usepackage[title]{appendix}
\theoremstyle{definition}
\newtheorem{theorem}{Theorem}
\newtheorem{proposition}[theorem]{Proposition}
\usepackage{color}
\def \dax #1 {{\color{red} [\textbf{Dax:} #1] }}
\def \blue #1 {{\color{blue}  #1 }}

\definecolor{burgundy}{rgb}{0.5, 0.0, 0.13}
\hypersetup{colorlinks,linkcolor={blue},citecolor={blue},urlcolor={burgundy}}  
\newcommand{\ketbra}[2]{\mbox{$|#1\rangle\langle #2|$}}

\newcommand{\myTriangleBlack}[4] % #1 = x coord, #2 = y coord, #3 = rotation angle, #4 = label
{\draw [rotate around={#3: (#1+.35, #2)}, fill=black] (#1, #2) +(0,.4) -- +(0,-.4) -- +(.7,0) -- +(0,.4) (#1+.25, #2) node[white] {#4}}

\begin{document}

\preprint{arXiv:1808.00460}

\title{Entanglement Scaling in Quantum Advantage Benchmarks}% Force line breaks with \\
%\thanks{Data and source code used to make Figure \ref{fig:gate-count} available online.}%
% supreme adversarial advantage 

\author{Jacob D.~Biamonte}
\email{j.biamonte@skoltech.ru}
 \homepage{http://quantum.skoltech.ru}
\affiliation{
 Deep Quantum Laboratory, Skolkovo Institute of Science and Technology, 3 Nobel Street, Moscow, Russia 121205
}

\author{Mauro E.S.~Morales}
\email{mauricioenrique.moralessoler@skoltech.ru}
\affiliation{
 Deep Quantum Laboratory, Skolkovo Institute of Science and Technology, 3 Nobel Street, Moscow, Russia 121205
}

\author{Dax Enshan Koh}
\email{daxkoh@mit.edu} 
\affiliation{%
Department of Mathematics, Massachusetts Institute of Technology, Cambridge, Massachusetts 02139, USA}%
\affiliation{
 Deep Quantum Laboratory, Skolkovo Institute of Science and Technology, 3 Nobel Street, Moscow, Russia 121205
}

%\date{\today}% It is always \today, today,
             %  but any date may be explicitly specified
\begin{abstract}
A contemporary technological milestone is to build a quantum device performing a computational task beyond the capability of any classical computer, an achievement known as quantum adversarial advantage. In what ways can the entanglement realized in such a demonstration be quantified?  Inspired by the area law of tensor networks, we derive an upper bound for the minimum random circuit depth needed to generate the maximal bipartite entanglement correlations between all problem variables (qubits).  This bound is (i) lattice geometry dependent and (ii) makes explicit a nuance implicit in other proposals with  physical consequence.  The hardware itself should be able to support super-logarithmic ebits of entanglement across some poly($n$) number of qubit-bipartitions, otherwise the quantum state itself will not possess volumetric entanglement scaling and full-lattice-range correlations.  Hence, as we present a connection between quantum advantage protocols and quantum entanglement, the entanglement implicitly generated by such protocols can be tested separately to further ascertain the validity of any quantum advantage claim. 
\end{abstract}
%\pacs{03.67.−a \and 03.67.Ac \and 03.67.Lx \and 02.20.Sv \and 02.30.Xx \and 02.30.Yy \and 03.67.Mn}

\maketitle

%\tableofcontents

%\section{Introduction}

Rapid experimental advancements have spawned an international race towards the first experimental {\it quantum adversarial advantage} demonstration---in which a quantum computer outperforms a classical one at some task \cite{ bremner2010classical, aaronson2011computational,preskill2012quantum, farhi2016quantum,2016arXiv160800263B, harrow2017quantum, bremner2017achieving, 2017arXiv171205384B,2017arXiv171005867P,Aaronson:2017:CFQ:3135595.3135617, 2018arXiv180404797L,2018arXiv180501450C, bouland_et_al:LIPIcs:2018:8867, bermejo2018architectures, bouland2018quantum, dalzell2018many,2018arXiv190710749, barends2016}. There is likewise  interest in understanding the effectiveness of low-depth quantum circuits for e.g.~machine learning \cite{2017Natur.549..195B} and quantum simulation \cite{bermejo2018architectures}.  Adding to this theory, we propose a quantification of the maximal entanglement (manifest in correlations between problem variables) that a given quantum computation can support \cite{2013JPhA...46U5301B}.  Quantification of the minimal-sized circuits needed to produce---even in principle---maximally correlated quantum states fills gaps missing in the theory of quantum adversarial advantage and low-depth circuits in general. Indeed, such minimal-depth circuits---as predicted by our theory---seem to be  difficult small quantum circuits to simulate classically and hence might offer quantum advantage for practical problems even outside of random adversarial advantage benchmarks.  

The goal of quantum adversarial advantage is to perform any task that is beyond the capability of any known classical computer. A naive starting point would be to consider the evident memory limitations of classical computers. And to create a quantum state exceeding that.  If we consider an ideal quantum state, we must store at most $2^{n+1}\cdot 16$ bytes of information, assuming 32 bit precision.  This upper bound reaches $80$ terabytes (TB) at just less than 43 qubits and $2.2$ petabytes (PB) at just under 47.  Eighty TB and $2.2$ PB are commonly referenced as the maximum memory storage capacity of a rapid supercomputing node and the supercomputer Trinity with the world's largest memory (respectively).  And so quantum adversarial advantage might already be possible with $\geq$ 47 qubits (strong simulation).  The problem is that to create states requiring $2^{n+1}$ independent degrees of freedom would require $O(\exp[n])$ gates, well beyond the coherence time of any device outside of the fault-tolerance threshold. And so we must search for another adversarial advantage protocol, requiring lower-depth circuits.

\paragraph*{Background.}
Broadly speaking, the leading proposals for quantum adversarial advantage can be divided into two categories: (i) those that provide strong complexity-theoretic evidence of classical intractability (based, for example, on the non-collapse of the polynomial hierarchy) and (ii) those that promise to be imminent candidates for experimental realization. Examples in the former category include sampling from (a) boson sampling circuits \cite{aaronson2011computational}, (b) IQP circuits \cite{bremner2010classical}, and (c) DQC1 circuits \cite{PhysRevLett.120.200502}. A leading example in the latter category is the problem of sampling from random quantum circuits. 

The existence of an efficient classical algorithm which can simulate random quantum circuits seems unlikely. In particular, it would imply the violation of the Quantum Threshold Assumption (QUATH) \cite{Aaronson:2017:CFQ:3135595.3135617}.  However, this says nothing of the number of qubits and the depths of the circuits required to demonstrate this separation between quantum and classical computational devices.  To address this, all arguments to-date have extrapolated---based on numerics or counting resources---where the classical intractability crossover point will occur.  Our theory does not avoid this per se, it simply provides upper bounds of the entanglement generated by a random circuit.  The same amount of entanglement generated by a random adversarial advantage circuit can be tested external to the adversarial advantage demonstration. We bound the required gate depth to maximally entangle any bipartition on a quantum processor, prior numerical simulations of the benchmark proposed in \cite{2016arXiv160800263B} are below this bound.  It is important to note that the gate depths required to maximally entangle all bipartitions is not a sufficient condition to achieve quantum adversarial advantage. Our methods to perform analysis on the entanglement are based on ideas from tensor networks. Alternative approaches have previously studied the role of entanglement in quantum algorithms such as Shor's algorithm \cite{orus2004,kendon2006}.

An observation of central importance is that existing quantum processors rely on qubits where the restriction is that these qubits interact on the 2D planar lattice.  In the long-term, the specific layout will be of less consequence.  However, for low-depth circuits a subtle implication is that the lattice embodies a small-world property, in which long-range correlations must be induced as a sequence of nearest neighbor operations.  Indeed, the Hilbert space describing the quantum processor can be entirely induced by a tensor network \cite{2015JSP...160.1389B} with the same underlying grid-geometry of the Hamiltonian governing the quantum processor itself.  Our bounds are formulated in this setting and are generally applicable across all current quantum adversarial advantage protocols.  

\paragraph*{Results.}
We consider a quantum processor with a geometry represented as a graph $G=(V,E)$. Nodes $v\in V$ are qubits and edges $e\in E$ are couplings which can be adjusted to create interactions between qubits (two body gates). At any point in the calculation, the state of the system is described by the wave function corresponding to the qubits of the quantum processor. This naturally gives rise to a projected entangled pair state (PEPS) tensor network \cite{orus-2014}. 

We will consider quantum circuits that are formed by acting on the interaction graphs appearing in Figure \ref{fig:grid-peps}. 

A quantum process is hence a space-time diagram codified by a triple of natural numbers $l \times m \times g$ where we assume $n = l\cdot m$ qubits enumerate the nodes of a rectangular lattice $Q$ and $g$ is the gate-depth of circuits acting on $Q$.  As will be seen, the variation over all circuits of depth at most $g$ acting on an $l \times m$ qubit grid lifts to a state-space.  Here the edges of $Q$ connect $2(\sqrt{n}-1)\sqrt{n}$ horizontal (otherwise vertical) nearest neighbor pairs---where $\sqrt{n}$ will be deformed later as to deviate from a perfect square and hence capture the rectangular structure of certain contemporary quantum information processors (see \ref{appendix:grid}).  We will fix a canonical basis found from iterating all possible binary values of the qubits positioned on the nodes of $Q$, which is given by the complex linear extension of the domain $\{0,1\}^l\times \{0,1\}^m$.  

This assignment lifts the internal legs of $Q$ to linear operators between external (qubit) nodes and hence fully defines our state-space.  Indeed, the grid structure induces a  dichotomy between tensors of (i) valence (3,1) and (ii) valence (4,1) where the first is of type $\mathbb{C}_\chi^{\otimes 3} \rightarrow \mathbb{C}_2$ and the second is $\mathbb{C}_\chi^{\otimes 4} \rightarrow \mathbb{C}_2$.  The parameter $\chi$ will be defined later as the internal bond dimension.  We note that the minimum edge cut bipartitioning $n$ qubits into two halves is $mincut(Q) = \sqrt{n}$, which will become a quantity of significance. A graphical example for the induced PEPS from a grid is shown in Fig. \ref{fig:grid-peps}.

Rank is the Schmidt number (the number of non-zero singular values) across any of the  bipartitions into $\lceil n/2 \rceil$ qubits on a grid.  Rank provides an upper-bound on the bipartite entanglement that a quantum state can support---as will be seen, a rank-$k$ state has at most $\log_2(k)$ ebits of entanglement. This provides an entanglement coarse-graining which we use to quantify circuits.  

An ebit is a unit of entanglement contained in a maximally entangled two-qubit (Bell) state. A quantum state with $q$ ebits of entanglement (quantified by any entanglement measure) contains the same amount of entanglement (in that measure) as $q$ Bell states. If a task requires $r$ ebits, it can be done with $r$ or more Bell states, but not with fewer.  Maximally entangled states in $\mathbb{C}^d\otimes \mathbb{C}^d$ have $\log_2(d)$ ebits of entanglement.  The question is then to upper bound the maximum amount of bipartite entanglement a given quantum computation can generate, turning to the aforementioned entanglement coarse-graining to classify quantum algorithms in terms of both the circuit depth, as well as the maximum ebits possible.  For low-depth circuits, these arguments are surprisingly relevant.  

To understand this, we note that the maximum number of ebits generated by a fully entangling two-qubit gate acting on a pair of qubits is never more than a single ebit.   We then consider that the maximum qubit partition with respect to ebits is into two (ideally) equal halves, which is never more than $\lceil n/2 \rceil$. We then arrive at the general result that a $g$-depth quantum circuit on $n$ qubits never applies more than $\min\{\lceil n/2 \rceil, g\}$ ebits of entanglement.  This in turn (see \ref{appendix:lb}) puts a lower-bound of $\log_2 \chi = \sqrt{n}/2$ on the two-qubit gate-depth to potentially drive a system into a state supporting the maximum possible ebits of entanglement. However, the grid structure requires the two-qubit gates acting on each qubit to be stacked, immediately arriving at $\sim \sqrt{4n}$ as the lower-bound for a circuit to even in principle generate $\lceil n/2 \rceil$ ebits of entanglement.  This lower bound is just below the gate-depths of interest which were successfully simulated in the literature (see Figure \ref{fig:gate-count} and the Discussion).  Under our coarse grained definition, we don't increase entanglement by the addition of local gates. Following the benchmark of \cite{2016arXiv160800263B}, local gates are added before and after each two qubit gate, multiplying the gate depth by a factor of two and the possible patterns of two qubit gates are eight which to take in consideration we multiply the gate depth by two, yielding $\sim 8\sqrt{n}$. This number bounds the gate depth of the circuit following the protocol in \cite{2016arXiv160800263B} to maximally entangle any bipartition in the grid. This bound is shown in Figure \ref{fig:gate-count}.  For further generalizations see \ref{appendix:lb}.

\begin{figure*}
\subfloat[A quantum processor of $3\times 3$ qubits]{\label{fig:grid}%
  \includegraphics[width=0.3\textwidth]{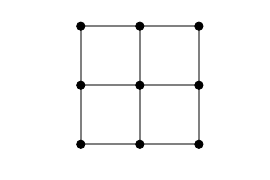}%
}
\subfloat[PEPS induced by grid in Fig. \ref{fig:grid}]{\label{fig:peps}%
  \includegraphics[width=0.3\textwidth]{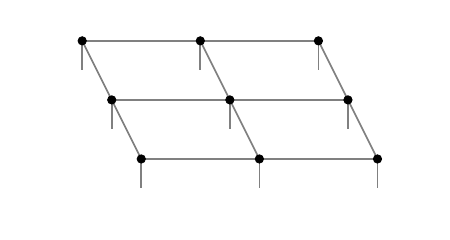}%
}
  \caption{Example of a grid quantum processor (left) and the induced projected entangled pair state (right). (a) Grid for a quantum processor of nine qubits. Each node is represents a qubit and each edge a tunable coupling between the qubits. (b) PEPS induced by grid in Fig.~\ref{fig:grid}. Each node with an open wire is a tensor, the open wires represent physical degrees of freedom in the tensor network and the internal edges represent internal bonds of some fixed dimension. }\label{fig:grid-peps}
\end{figure*}

\paragraph*{Discussion.}
Figure \ref{fig:gate-count} compares our bound to quantum adversarial advantage predictions. Data-points included follow the prescription of quantum circuit simulation by Google \cite{2016arXiv160800263B}. The gate set used in this prescription comprises: $\mathrm{H}$, $\sqrt{\mathrm{X}}$, $\sqrt{\mathrm{Y}}$, $\mathrm{T}$, $\mathrm{CZ}$. Gate definitions are given in Table \ref{table:gates}.

While some of these simulations (e.g.~those done in Sunway TaihuLight supercomputer \cite{2018arXiv180404797L} for a depth $39$ circuit) involve the calculation of the amplitudes of all output bitstrings (all $2^{46}$ bit strings, in the case of Sunway TaihuLight), others such as Alibaba \cite{2017arXiv171005867P} or Sunway TaihuLight simulation for a depth $55$ circuit involve only the calculation of a single amplitude. The data points were obtained from different simulations done recently \cite{2016arXiv160800263B, 2017arXiv171205384B, 2018arXiv180501450C, 2018arXiv180404797L, 2017arXiv171005867P, 2018arXiv180206952C, Haner:2017:PSQ:3126908.3126947}. 

It is interesting to note that the reported numerical simulations fall below the $n/2$ ebit bound (pink online) depicted in Figure \ref{fig:gate-count}. Such circuits are considered to be difficult low-depth circuits to simulate classically.  

We have also included a heat map with an estimation for the running time based on state-of-art algorithms from Alibaba \cite{2017arXiv171005867P}. To estimate the running time, we made use of the following upper bound by 
Markov and Shi \cite{markov2008simulating}: 
any $\alpha$-local interacting quantum circuit of size $M$ and depth $g$ can be strongly simulated in time $t(M,g) = 10^{-17} \cdot M^{O(1)} \exp[O(\alpha g)]$, where a factor of $10^{-17}$ has been included so that the running time of the simulation is in units of seconds. For this factor, we assumed that a classical computer is capable of doing $10^{17}$ flops. In the case of our tensor network $G$ representing a $\sqrt{n} \times \sqrt{n}$ grid, $\alpha= \sqrt{n}$, since the quantum circuit that $G$ represents is $\sqrt{n}$-local interacting. We also estimate the number of total gates naively as the number of couplers in the grid multiplied by the depth gate $M(n,g) =  2(\sqrt{n}-1)\sqrt{n}g$.  Hence, we consider the equation 
\begin{equation} \label{eq:markov_fit}
\begin{split}
 t(n,g) &= 10^{-17} \cdot M(n,g)^{a_1} 2^{a_2 g \sqrt{n}} \\
       &= 10^{-17} \cdot [2(\sqrt{n}-1)\sqrt{n}g]^{a_1} 2^{a_2 g \sqrt{n}}
\end{split}
\end{equation}
and fit it to the numerical results of Alibaba \cite{2017arXiv171005867P} (where it has been taken into consideration the fact that only simulations for one amplitude were realized).
For our fit, we obtained the parameters $a_1 = 4.36063901$ and $a_2 = 0.04315488$. With this fit we are able to give an estimation for the gate depth that can be simulated in $1$ month, $1$ year, $10$ years and $100$ years. An important remark to make here is that Alibaba simulations only calculate $1$ amplitude of an exponential number of possible strings. The algorithm is a modification from Boixo et al.~\cite{2017arXiv171205384B} and based on treewidth to measure contraction complexity  as shown by Markov and Shi. Thus, this estimation should be considered as an approximation. Lastly, we include a pair of vertical lines corresponding to the quantum computers built by IBM and Google with 50 and 72 qubits, respectively. The estimations for achievable gate depths in classical computers for a given threshold are shown in Table \ref{table:runtimes}.

While preparing this manuscript, work by Markov et al.~appeared \cite{2018arXiv190710749} where the prescription on how gates are applied has been changed. One of these changes is the inclusion of the $\mathrm{iSwap}$ gate. They estimate that the gate depth to be simulated in a given runtime is about half for the state-of-art algorithms in this new benchmark. We don't include simulations of this latest work in Figure \ref{fig:gate-count}. Considering this, we show how the estimations are modified in Table \ref{table:runtimes_modbench}. 

In conclusion, we observe a nonlinear tradeoff between the number of qubits and gate depth, with the fleeting resource (with exponential dependency) being the gate depth. As is predicted, we remark that quantum adversarial advantage demonstrations (assuming completely random circuit families) should involve circuits of depth at least 50---if not more---for 80 to 150 qubits. Circuits of such depth would be inside the area in pink in Figure~\ref{fig:gate-count} that we derived in the study. We hence provide some handle to understand the physical entanglement that such adversarial advantage protocols aim generate.

\begin{acknowledgments}
We thank Igor Markov and Mark Saffman for useful comments and Sergio Boixo for insightful discussion regarding this work. Data and source code used to make Figure \ref{fig:gate-count} available online \cite{github}. 
\end{acknowledgments} 

\begin{figure*}
    \centering
    \includegraphics[width=\textwidth]{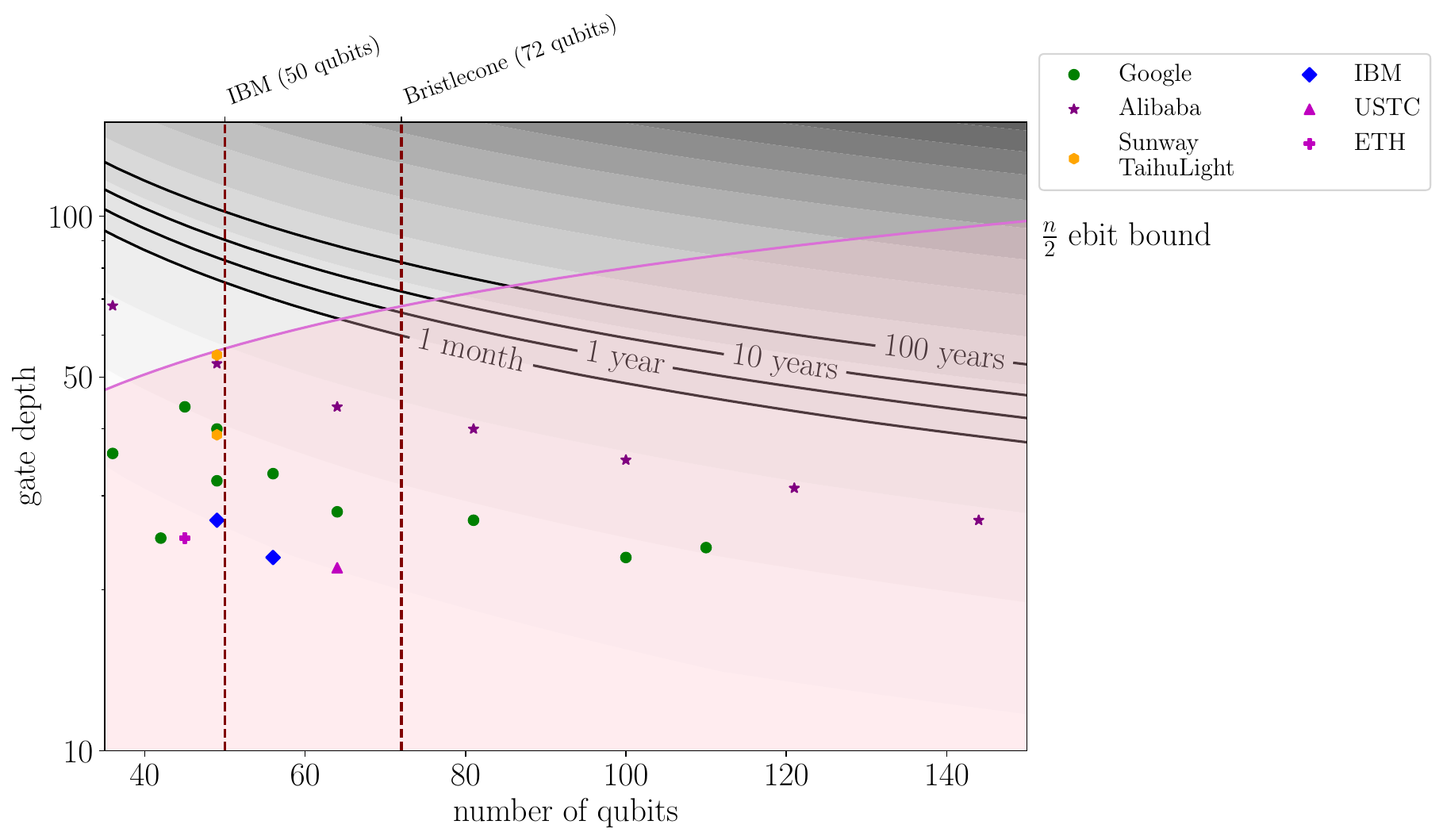} 
    \caption{Comparison of findings. Qubits versus gate depths superimposed on runtimes. In pink (online), area below the gate depth required to achieve maximallly entangled states for every bipartition containing existing numerical data (depicted as $n/2$ ebit bound). %Gate depth estimation to maximize entanglement across a minimum cut bipartition with equal number of qubits as a function of the total number of qubits in the lattice. This gate depth lower bound (blue curve) is multiplied by factors $4$ to account for local gates given the prescription from Google \cite{2016arXiv160800263B} with the gate set $\mathrm{CZ}$, $\mathrm T$, $\sqrt{\mathrm X}$, $\sqrt{\mathrm Y}$, $\mathrm H$. Existing numerical data fall above the analytical lower bound yet still under that bound with the addition of local single qubit gates.  
    [Data points from Google \cite{2016arXiv160800263B, 2017arXiv171205384B}, Alibaba \cite{2018arXiv180501450C}, Sunway TaihuLight \cite{2018arXiv180404797L}, IBM \cite{2017arXiv171005867P}, USTC \cite{2018arXiv180206952C}, ETH \cite{Haner:2017:PSQ:3126908.3126947}].}
    \label{fig:gate-count}
\end{figure*}

%\onecolumn
\bibliographystyle{naturemag}  
\onecolumngrid
\bibliography{sample}

\begin{appendices}

\section{Lower bounds for general tensor networks}\label{appendix:lb}

In the main text, we derived bounds for qubits positioned on a grid. In this section, we generalize the result to general graphs. Consider a tensor network $G$ with $n$ vertices and $e$ edges. For convenience, we assume that $n$ is even. Assume that the bond dimension associated with each edge is equal to $\chi$. Let $C$ be a cut of $G$ that partitions the vertices of $G$ into two disjoint subsets, each of equal size $n/2$. Let $f$ be the number of edges that have exactly one endpoint in each of the two parts of $C$.

\begin{proposition}
\label{prop:generalGraph}
 If $G$ supports the maximum number of ebits across $C$, then
    $$\chi \geq 2^{\frac{n}{2f}}. $$
 The total number of gates $\#_g$ needed to produce the state represented by $G$ satisfies
    $$\#_g \geq \frac{ne}{2f}. $$
\end{proposition}
\begin{proof}
The maximum number of ebits across $C$ is $n/2$. Hence,
$$   2^{n/2} \leq \chi^f, $$
which implies that 
$$ \chi \geq 2^{\frac{n}{2f}}.$$

The total number of gates associated with each edge is $\log_2 \chi \geq \frac{n}{2f}$. Multiplying this quantity with the number of edges $e$, the total number of gates needed to produce the state represented by $G$ satisfies
\begin{equation}\label{eq:total-gates}
    \#_g = e \log_2 \chi \geq \frac{ne}{2f}. 
\end{equation}

\end{proof} 

In the case where $G$ is a $\sqrt n\times \sqrt n$ grid, we have $e=2\sqrt n (\sqrt n - 1)$. If $C$ is a min-cut of $G$, then $f = \sqrt n$. Substituting these values of $e$ and $f$ into Proposition \ref{prop:generalGraph} allows us to recover the equations obtained in the main text:
    $$\chi \geq 2^{\sqrt n/2}, $$
and
    $$\#_g \geq n(\sqrt n -1). $$

Several other geometries can be analyzed as well. Consider a quantum processor with a line and a ring topology such as in Figure \ref{fig:line} and \ref{fig:ring}. Following Eq.~\eqref{eq:total-gates}, the total number of gates $\#_g$ to maximally entangle any bipartition in a quantum processor described by a graph $G$  must be at least $\frac{ne}{2f}$ where $n$ is the number of vertices, $e$ is the total number of edges and $f$ is the number of edges that are cut by the bipartition cutting the most edges in the graph (but keeping the bipartition with equal number of vertices).
For the line we have that given $n$ vertices then $e=n-1$ and $f=1$, on the other hand for the ring topology $e=n$ and $f=2$. We thus conclude that

$$\#_g^{line} \geq\frac{n(n-1)}{2}. $$
$$\#_g^{ring}\geq\frac{n^2}{4}. $$
showing that for achieving a maximally entangled bipartition requires more gates in the line topology than in the ring one for large $n$.

We consider as a last example a topology proposed by Rigetti \cite{rigetti2019} for a 128-qubit processor. The graph consists of octagons such as the one in Figure \ref{fig:ring} which form a grid as in Figure \ref{fig:rigetti}. Call $o$ the number of octagons in the graph, then $e=8o + 2 \times (2\sqrt{o}(\sqrt{o}-1) )$. The first term comes from the fact that each octagon has $8$ edges, the second comes from the fact that there are $(2\sqrt{o}(\sqrt{o}-1) )$ double edges between octagons in the grid of octagons. It is direct the the fewest number of edges that a bipartition can cross is $f=2\sqrt{o}$. Since $o=\frac{n}{8}$, we have that

$$\#_g^{Rigetti} \geq \frac{ne}{2f} = \frac{3}{\sqrt{8}}n\sqrt{n}-n. $$
Which implies that the required number of gates applied in the Rigetti topology is greater than in the grid case as it would be expected.

\begin{figure}
\subfloat[Line topology]{\includegraphics[width=.3\columnwidth]{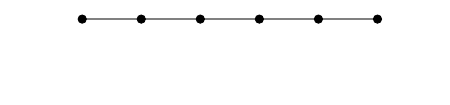}\label{fig:line} }
\qquad
\subfloat[Ring topology]{\includegraphics[width=.3\columnwidth]{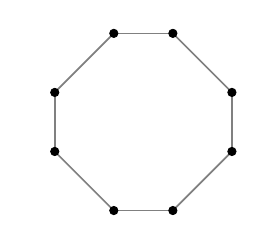}\label{fig:ring} }%
\qquad
\subfloat[Rigetti topology]{\includegraphics[width=.3\columnwidth]{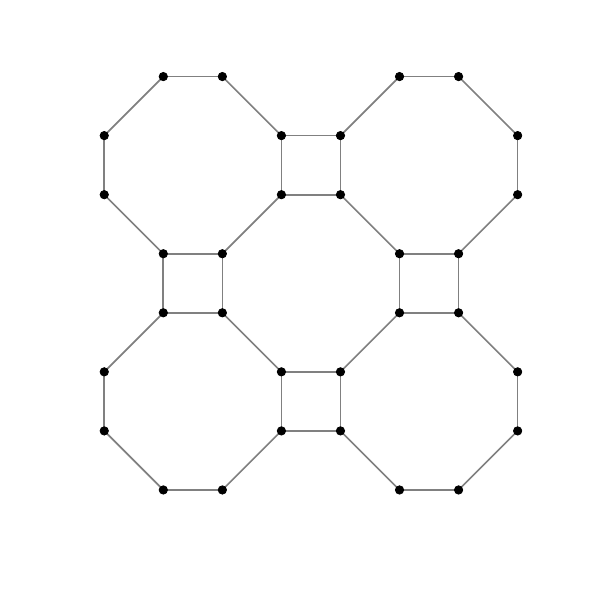}\label{fig:rigetti} }
\caption{Examples of graph topologies for a quantum processor. (a) A line with $n=6$ qubits. (b) A ring with $n=8$ qubits. (c) Topology of a quantum processor from architecture announced by Rigetti \cite{rigetti2019}, we consider here a grid composed of octagonal rings connected as shown in the figure.}%
\label{fig:line-ring}
\end{figure}

\section{Lower bounds for deformed grid}\label{appendix:grid}
Let $H$ be a $(\sqrt n-k) \times \left(\frac n{\sqrt n - k}\right)$ grid. In other words, $H$ is a deformation of the 
$\sqrt n\times \sqrt n$ grid $G$, wherein the number of qubits along one edge is reduced by $k$ and the total number of qubits is kept constant at $n$.
Assume that $\sqrt n > k \geq 0$, that $n$ is a perfect square, and that $\sqrt{n}-k$ divides $n$, to ensure that the number of qubits in each row and column is an integer.

Note that the length of the longer edge can be written as
$$
    \frac n{\sqrt n -k} = \sqrt n + k + \frac{k^2}{\sqrt n} + O\left(\frac{k^3}n\right).
$$

Hence, for $H$, the number of vertices is $n$, and the number of edges is
$$
    e = 2 \sqrt n(\sqrt n -1) - \frac {k^2}{\sqrt n} \left(1-\frac k{\sqrt n}\right)^{-1}.
$$

If $C$ is the min-cut of $G$, then $f = \sqrt n-k$. Applying Proposition \ref{prop:generalGraph} to the graph $H$ and the cut $C$ then gives
$$
    \log \chi \geq \frac{\sqrt n}2 \left( 1-\frac k {\sqrt n}\right)^{-1},
$$
and
$$
\#_g \geq n(\sqrt{n} -1)\left( 1 - \frac{k}{\sqrt n}\right)^{-1} - \frac{k^2}{2}\left( 1 - \frac{k}{\sqrt n}\right)^{-2}.
$$

Since $\frac{k}{\sqrt n} < 1$, we could expand $\left( 1 - \frac{k}{\sqrt n}\right)^{-2}$ and $\left( 1 - \frac{k}{\sqrt n}\right)^{-2}$ as Taylor series. This gives
$$
\log \chi \geq \frac{\sqrt n}2  \sum_{s=0}^\infty \left(\frac k{\sqrt n}\right)^s 
= \frac{ \sqrt n}2 \left[1+ \frac k{\sqrt n} + \frac{k^2}n + O\left(\left(\frac{k}{\sqrt n}\right)^3\right)\right],
$$
and
$$
\#_g \geq \sum_{s=0}^\infty \left[ n(\sqrt n-1) - \frac{k^2} 2 (s+1) \right] \left(\frac k{\sqrt n}\right)^s.
$$

% \section{Notes for plot}\label{appendix:notes-plot}
% Here we include a description of how the running times in Figure \ref{fig:gate-count} were obtained. Starting from \eqref{eq:markov_fit}, we can fit the runtimes to the ones reported by Alibaba \cite{2017arXiv171005867P}. From the $7$ data points reported by Alibaba, the point corresponding to $144$ qubits was not considered for the fit because of its significant lower runtime compared to the other data points. Note that the runtime in \eqref{eq:markov_fit} is in units of seconds. The parameters obtained for this fit are $a_1 = 4.36063901$ and $a_2 = 0.04315488$. 
% We used this fit to estimate the achievable gate depths given different runtimes for grids of $50$ and $72$ qubits. The gate depths obtained are shown in Table \ref{table:runtimes}. 

% While preparing this manuscript, work from Google appeared \cite{2018arXiv190710749} where the prescription on how gates are applied has been changed. One of these changes is the inclusion of the $\mathrm{iSwap}$ gate. They estimate that the gate depth to be simulated in a given runtime is about half for the state-of-art algorithms in this new benchmark. Considering this changes and the estimation given by Google, then the gate depths are modified as shown in Table \ref{table:runtimes_modbench}.
 \newpage

\FloatBarrier

\setlength{\tabcolsep}{20pt}
\begin{table}[H]
\centering
\begin{tabular}{ c c  } 
 \hline
 \hline
$\mathrm{H}$ & $\frac{1}{\sqrt{2}}\sum_{k,l = 0}^1 (-1)^{kl} \ketbra{k}{l}$ \\
$\sqrt{\mathrm{X}}$ &  $\frac{1}{2}\big[(1+i) \ketbra{0}{0} + (1-i)\ketbra{0}{1} + (1-i)\ketbra{1}{0} + (1+i) \ketbra{1}{1} \big]$ \\
$\sqrt{\mathrm{Y}}$ &   $\frac{1}{2}\big[(1+i) \ketbra{0}{0} + (-1-i)\ketbra{0}{1} + (1+i)\ketbra{1}{0} + (1+i) \ketbra{1}{1} \big]$ \\
$\mathrm{T}$ & $\sum_{k = 0}^1 \mathrm{e}^{i\frac{\pi}{4}k} \ketbra{k}{k}$ \\
$\mathrm{CZ}$ & $\sum_{k,l=0}^1 (-1)^{kl} \ketbra{kl}{kl} $ \\
 %\hline 
 \hline
 \hline
\end{tabular}
 \caption{Gate set involved in the random circuit sampling benchmark recently proposed \cite{2016arXiv160800263B}.}
 \label{table:gates}
\end{table}

\setlength{\tabcolsep}{20pt}
\begin{table}[H]
\centering
\begin{tabular}{ c c c } 
 \hline
 \hline
 runtime & gate depth for $50$ qubits & gate depth for 72 qubits  \\ 
 \hline
 $1$ month & 75 & 60  \\ 
 %\hline 
 $1$ year & 84 & 67 \\ 
 %\hline
 $10$ years & 93 & 75  \\
 %\hline
 $100$ years & 102 & 82 \\
 \hline
 \hline
\end{tabular}
 \caption{Estimation based on the fit given by \eqref{eq:markov_fit} for the gate depths achievable in given runtime assuming strong classical simulation of one amplitude. For our estimates, we assume that 1 year contains $365.2425 \times 24 \times 60^2$ seconds, and that 1 month is $\tfrac 1{12}$ of a year. The gate depths obtained are rounded up to the nearest integer.}
 \label{table:runtimes}
\end{table}
%\endgroup

\setlength{\tabcolsep}{20pt}
\begin{table}[H]
\centering
\begin{tabular}{ c c c } 
 \hline
 \hline
 runtime & gate depth for $50$ qubits & gate depth for 72 qubits  \\ 
 \hline
 $1$ month & 38 & 30  \\ 
 %\hline 
 $1$ year & 42 & 33 \\ 
 %\hline
 $10$ years & 46 & 38  \\
 %\hline
 $100$ years & 51 & 41 \\
 \hline
 \hline
\end{tabular}
 \caption{Estimation based on the fit given by \eqref{eq:markov_fit} and estimation given by Google \cite{2018arXiv190710749} considering the change of benchmark for the gate depths achievable in given runtime assuming strong classical simulation of one amplitude. }
 \label{table:runtimes_modbench}
\end{table}

 \end{appendices}

\end{document}